\documentclass[onecolumn,english,reprint, longbibliography, superscriptaddress, breaklinks=true, showkeys, showpacs=false, nofootinbib]{revtex4}
\usepackage[T1]{fontenc}
\usepackage[latin9]{inputenc}
\usepackage{color}
\usepackage{babel}
\usepackage{amsmath}
\usepackage{amsthm}
\usepackage{amssymb}
\usepackage{subfigure}
\usepackage{graphicx}
\usepackage{physics} 
\usepackage{hyperref}
\hypersetup{
    colorlinks=true,
    linkcolor=blue,
    filecolor=blue,      
    urlcolor=blue,
    }
\makeatletter
\@ifundefined{textcolor}{}
{%
 \definecolor{BLACK}{gray}{0}
 \definecolor{WHITE}{gray}{1}
 \definecolor{RED}{rgb}{1,0,0}
 \definecolor{GREEN}{rgb}{0,1,0}
 \definecolor{BLUE}{rgb}{0,0,1}
 \definecolor{CYAN}{cmyk}{1,0,0,0}
 \definecolor{MAGENTA}{cmyk}{0,1,0,0}
 \definecolor{YELLOW}{cmyk}{0,0,1,0}
}
\pdfoutput=1
\hypersetup{colorlinks=true,citecolor=blue,linkcolor=cyan,urlcolor=blue,filecolor= green, breaklinks=true}
\usepackage{url}
\usepackage{breakurl}
\makeatother
\newtheorem{theorem}{Theorem}
\newtheorem{corollary}{Corollary}
\newtheorem{definition}{Definition}

\begin{document}

\title{Restricting to the chip architecture maintains the quantum neural network accuracy
}

\author{Lucas Friedrich}
\email[Electronic address: ]{lucas.friedrich@acad.ufsm.br}
\affiliation{Physics Departament, Center for Natural and Exact Sciences, Federal University of Santa Maria, Roraima Avenue 1000, 97105-900, Santa Maria, RS, Brazil}

\author{Jonas Maziero}
\email[Electronic address (corresponding author): ]{jonas.maziero@ufsm.br}
\affiliation{Physics Departament, Center for Natural and Exact Sciences, Federal University of Santa Maria, Roraima Avenue 1000, 97105-900, Santa Maria, RS, Brazil}

\selectlanguage{english}%

\begin{abstract}

In the era of noisy intermediate-scale quantum devices, variational quantum algorithms (VQAs) stand as a prominent strategy for constructing quantum machine learning models. These models comprise both a quantum and a classical component. The quantum facet is characterized by a parametrization \(U\), typically derived from the composition of various quantum gates. On the other hand, the classical component involves an optimizer that adjusts the parameters of \(U\) to minimize a cost function \(C\). Despite the extensive applications of VQAs, several critical questions persist, such as determining the optimal gate sequence, devising efficient parameter optimization strategies, selecting appropriate cost functions, and understanding the influence of quantum chip architectures on the final results.
This article aims to address the last question, emphasizing that, in general, the cost function tends to converge towards an average value as the utilized parameterization approaches a \(2\)-design. Consequently, when the parameterization closely aligns with a \(2\)-design, the quantum neural network model's outcome becomes less dependent on the specific parametrization. This insight leads to the possibility of leveraging the inherent architecture of quantum chips to define the parametrization for VQAs. By doing so, the need for additional swap gates is mitigated, consequently reducing the depth of VQAs and minimizing associated errors.
\end{abstract}

\keywords{Quantum neural networks; Variational quantum algorithms; Quantum chip architecture}

\maketitle

\section{Introduction}

Quantum machine learning is an emerging interdisciplinary area of study that involves quantum computing and machine learning \cite{biamonte,schuld}. This investigation path seeks to create a quantum machine learning model with greater computational power than its classical counterparts. For this, phenomena that are only possible to obtain in the quantum domain are used, such as entanglement and superposition. Currently, in the era of noisy intermediate-scale quantum devices, variational quantum algorithms have proved to be one of the best strategies for building quantum machine learning models. Variational quantum algorithms (VQAs) \cite{cerezo,tilly} are built using a quantum part and a classical part, as illustrated in Fig. \ref{fig:vqaFig}. The quantum part is obtained from a parameterization $U(\theta)$ composed of different quantum gates and with parameters $\theta$. In turn, the classical part refers to an optimizer whose objective is to minimize a cost function $C$. For this, the classical optimizer iteractively updates the parameters $\pmb{\theta}$ of the parameterization $U$. 

Although the number of works in this area has grown in recent years, there are still many questions and problems to be resolved. One example is the vanishing gradient, which is also known as barren plateaus \cite{BR_cost_Dependent,BR_Entanglement_devised_barren_plateau_mitigation,BR_Entanglement_induced_barren_plateaus,BR_expressibility,BR_noise,BR_gradientFree}. The problem of barren plateaus is related to the difficulty of optimizing the parameters as the system size increases. In general, in VQAs the optimization of these parameters is performed iteractively using the following rule

\begin{equation}
    \pmb{\theta}^{t+1} = \pmb{\theta}^{t} - \eta \nabla_{\pmb{\theta}}C.\label{eq:gradiente_rule}
\end{equation}

This rule tells us that for a given epoch $t+1$, the new parameters will be obtained from the difference between the parameters of the previous epoch $t$ and the gradient of the cost function. Here the term $\eta$ is called learning rate and its goal is to control the size of the step that will be taken. That is, it will weigh how much the gradient value will influence the new values of $\pmb{\theta}$. However, due to the problem of barren plateaus, we have $ \nabla_{\pmb{\theta}}C \approx 2^{-n} $, where $n$ is the number of qubits in the VQA. This means that the gradient will tend exponentially to zero as the number of qubits used grows. As a consequence, as the system size increases, the computational cost related to parameter optimization increases exponentially. Furthermore, in Ref. \cite{BR_gradientFree} it was shown that this problem is also present in gradient-free methods, that is, optimization methods where the gradient of the cost function is not used. Some methods to mitigate this problem have already been suggested \cite{friedrich,BR_initialization_strategy,BR_Large_gradients_via_correlation,BR_LSTM,BR_layer_by_layer}, but this is still an open area of study.

\begin{figure}
    \centering
    \includegraphics[scale=0.5]{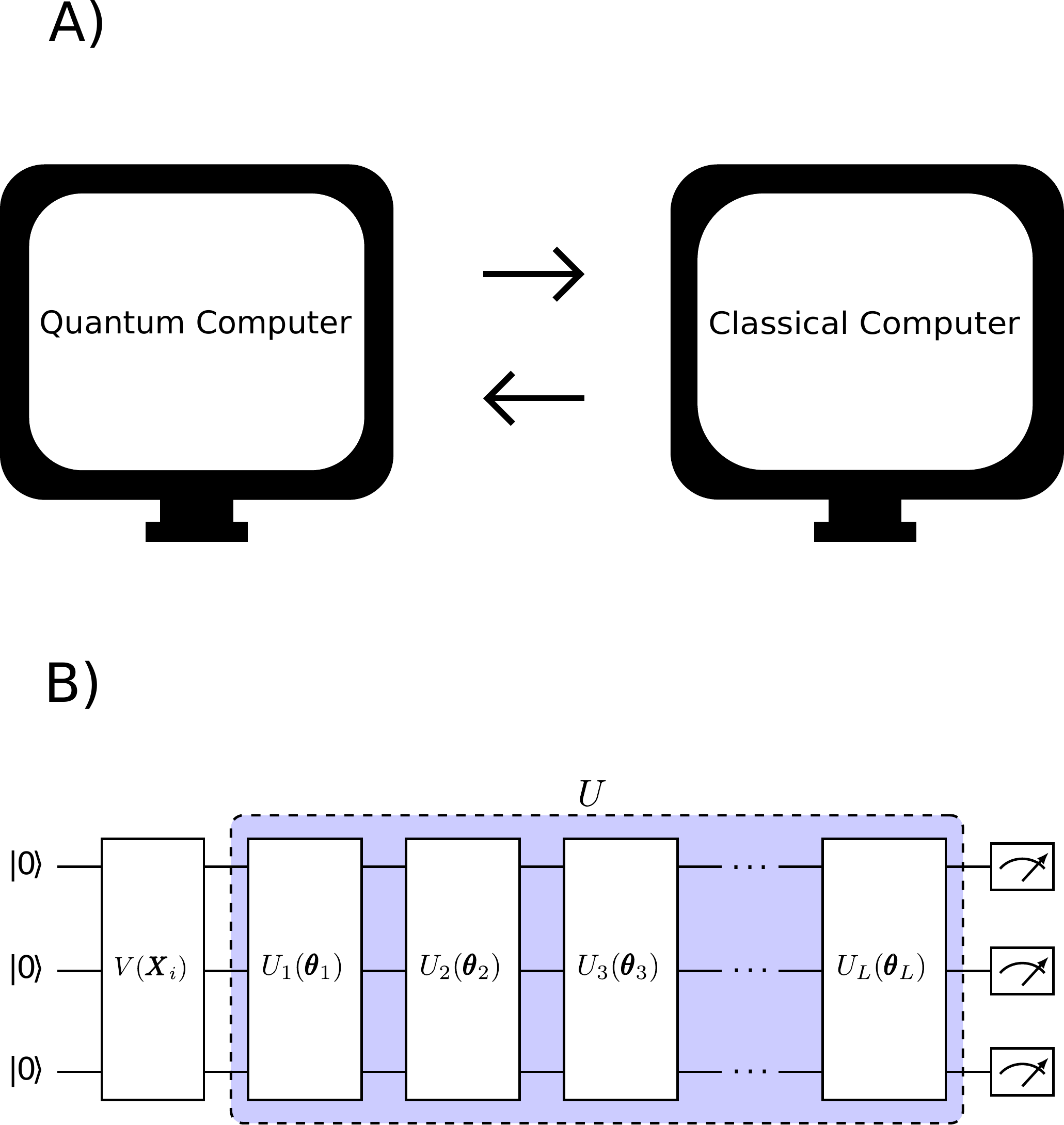}
    \caption{A) Illustration of the operation of variational quantum algorithms. These models work in a hybrid way, that is, using a quantum computer and a classical computer. In the quantum computer the quantum circuit will be run and in the classical computer the classical optimizer will be run. B) Illustration of a quantum circuit with three qubits.}
    \label{fig:vqaFig}
\end{figure}

Furthermore, problems such as parameters optimization \cite{Friedrich_nes,Rebentrost,Schuld_Evaluating}, influence of parameterization expressibility \cite{Expressibility_Sim,Hubregtsen,expr_friedrich} and influence of the architecture of the quantum chips \cite{nash,bravyi} still have to be addressed, even though some progress in this direction have already been made \cite{kandala,benedetti,nguyen,du}.


\begin{figure}[b]
    \centering
    \includegraphics[scale=0.45]{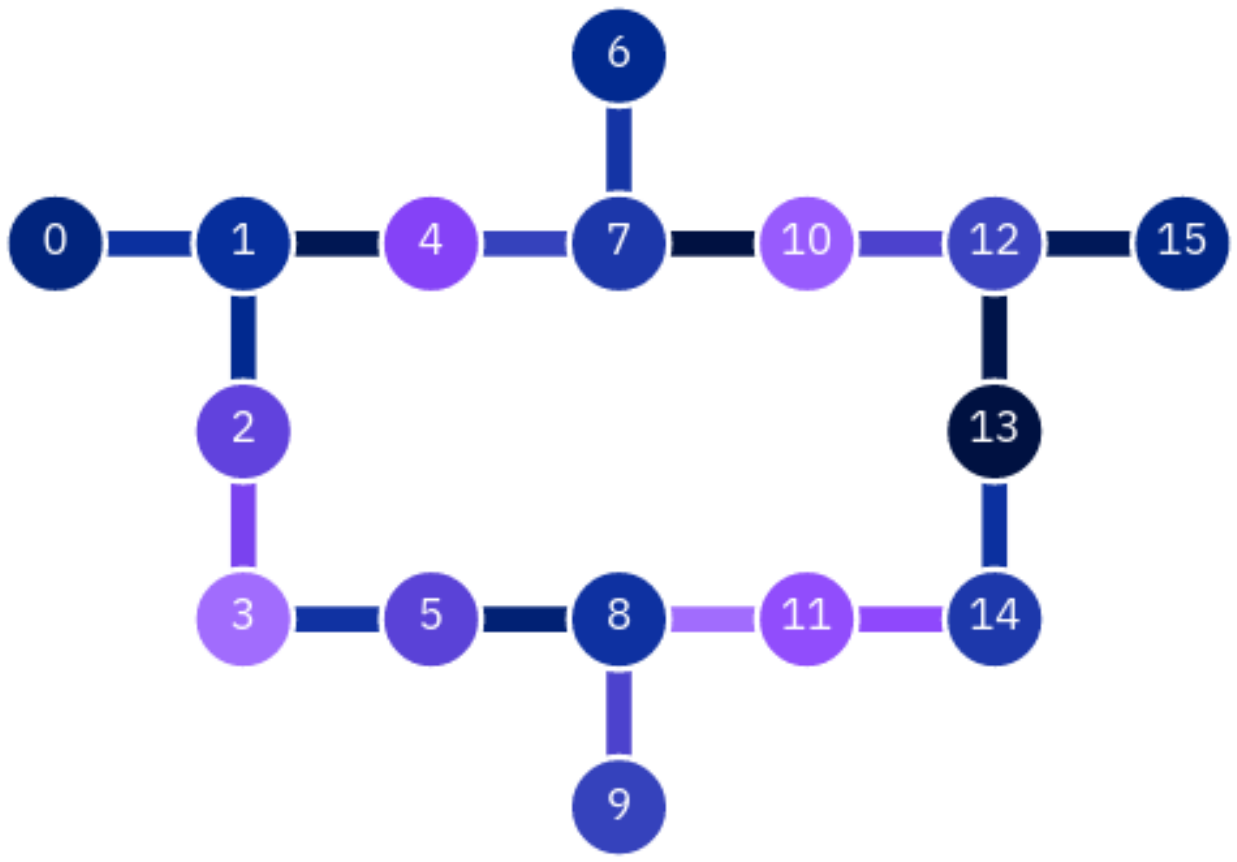}
    \caption{Illustration of IBM's Guadalupe chip connectivity architecture.}
    \label{fig:chip_guadalupe}
\end{figure}

In this article, we study how the quantum computer chip architecture influences the efficiency of a quantum neural network model. More precisely, we will analyze how the connectivity of superconducting chips affects the results obtained. 
For example, IBM quantum processors
have different chip architectures. This architecture determines which qubit pairs interact directly.
For example, in Fig. \ref{fig:chip_guadalupe} is shown one of the architectures made available by IBM. In this figure, each circle represents a qubit with its respective index, which in this case varies from 0 to 15, and the edges represent the connectivity between the qubits. Therefore, we can see that not all qubits are directly connected. For example, qubit 0 is directly connected only with qubit 1. As a consequence of this limited connectivity, if we wish, for example, to apply a CNOT gate between qubits 0 and 2, we must initially apply a SWAP gate between qubits 1 and 2 and only after that we can apply the CNOT gate between qubits 0 and 2. After that we must apply the SWAP gate again between qubits 1 and 2. Therefore, as a SWAP gate is implemented using three CNOT gates, the depth of the parameterization $U$ will increase, leading to serious problems mainly due to the noise present in current chips. For avoiding this problem, an intuitive approach is to consider parameterizations $U$ involving CNOT gates applied only to pairs of qubits that are directly connect in the quantum chip.
But this raises the question about how this will affect the efficiency of the corresponding VQA. This is the question that we investigate in this article.

The remainder of this article is organized as follows. In Sec. \ref{sec:qnn}, we describe the three main stages of a quantum neural netwok, i.e., the data encoding, parametrization and measurement stages. In Sec. \ref{sec:thm}, we present our theoretical results and briefly discuss their implications. 
In Sec. \ref{sec:sim}, we give details about our simulation method. With this, in Sec. \ref{sec:results}, we present the results obtained in our simulations. Finally, in Sec. \ref{sec:conc}, we give our concluding remarks.

\section{Quantum neural networks}
\label{sec:qnn}

Quantum neural networks are neural network models inspired by classical models. These models can be divided into two main parts, a quantum and a classical one. The quantum part is constructed as follows. Initially we use a $V$ parameterization whose objective is to prepare a quantum state. In machine learning, especially supervised learning, the goal of the model is to learn patterns from a set of training data. Thus, given a training data set $\mathcal{D} := \{ \pmb{x}_{i}, y_{i} \}_{i=1}^{D}$, where $\pmb{x}_{i}$ is the input data, for example, images of handwritten digits, the model must learn to obtain the respective labels $y_{i}$. In quantum neural networks we use the $V$ parameterization to map the data $\pmb{x}_{i}$ into a quantum state. Some examples of data coding are wave function encoding
\begin{equation}
  |\pmb{x} \rangle  := \frac{1}{\Vert \pmb{x} \Vert_{2}^{2} }\sum_{i=1}^{2^{N}}x_{i}|i \rangle,
  \label{eq:wave_function_encoder}
\end{equation}
the dense angle encoding
\begin{equation}
   | \pmb{x} \rangle = \bigotimes_{i=1}^{N/2}\big(\cos( \pi x_{2i-1} )|0\rangle  + e^{2\pi i x_{2i}}\sin( \pi x_{2i-1} )|1\rangle\big),  \label{eq:dense_angle_coding}
\end{equation}
and the qubit encoding
\begin{equation}
    |\pmb{x}\rangle = \bigotimes_{i=1}^{N}\big( \cos(x_{i})|0\rangle + \sin(x_{i})|1\rangle\big).
    \label{eq:qubit_encoding}
\end{equation}
The choice of this parameterization, as shown in Refs. \cite{expressive_SCHULD,Data_re_uploading_Salinas}, significantly affects the model and the results obtained. For example, in Ref. \cite{Data_re_uploading_Salinas} it was shown numerically that if we reload the data between the parameterization layers, the results obtained by the model are superior to those obtained when the data is loaded only at the beginning of the circuit.
After encoding the input data into quantum states, we apply a parametrization $U(\pmb{\theta})$. This parametrization is generally obtained from the composition of different quantum gates, such as CNOT, SWAP, rotation gates, etc. In a quantum neural network model, this part is the analogue of the hidden layers used in a classical neural network model. In general, we can write the parametrization as follows
\begin{equation}
    U(\pmb{\theta}) = \prod_{l=1}^{L} U_{l}W_{l},
    \label{eq:parametrization_1}
\end{equation}
where $L$ is the parametrization depth, $U_{l}$ is an arbitrary parametrization and $W_{l}$ is a ``fixed'' parameterization, that does not depend on any parameters. For example, this parametrization can be obtained from the product of applied CNOT gates in the qubits. For the purposes of this article, we define the parametrization such that
\begin{equation}
    U_{l} = \bigotimes_{j=1}^{N} R_{\sigma}(\theta_{l,j}),
    \label{eq:parametrization_2}
\end{equation}
where $ R_{\sigma}(\theta_{l,j}) = e^{-i \theta_{l,j} \sigma/2} $ with $\sigma \in \{ \sigma_{x},\sigma_{y},\sigma_{z} \} $ being one of the Pauli matrices.

Finally, the third part that makes up the quantum circuit is the measurement. There are different ways of defining these measurements, and their choice has great relevance to the problem. For example, in Ref. \cite{BR_cost_Dependent} it was shown that the phenomenon of barren plateaus is related to the choice of these measurements. However, in general, we can define the measurement as
\begin{equation}
    f(x,U) = Tr[U\rho U^{\dagger}H],
    \label{eq:medida}
\end{equation}
where $\rho := |\pmb{x}\rangle \langle \pmb{x}| $ and $H$ is an observable. For instance, we can use $H := |i\rangle \langle i|$, where $|i\rangle$ is one of the computational base states. Or we can utilize $H := \bigotimes_{i=1}^ {N} \sigma_{i} $ with $\sigma_{i}$ being one of the Pauli matrices applied to the qubit indexed $i$ and $N$ is the number of qubits used.


The classical part of a quantum neural network consists of a classical optimizer that uses a cost function $C$, that is defined in general as
\begin{equation}
    C = \frac{1}{D}\sum_{i=1}^{D}l( f(\pmb{x}_{i},y_{i}) ),
\end{equation}
where $f(\pmb{x}_{i})$ is the label predicted by the model, $y_{i}$ is the target label and $l(.,.)$ is a function that quantifies the difference between $f(\pmb{x}_{i})$ and $y_{i}$.
Different optimization proposals have already been made \cite{Friedrich_nes,Rebentrost,Schuld_Evaluating}. However, as in classical neural networks, this optimization is generally done using the rule given by Eq. \eqref{eq:gradiente_rule}.
Thus, the training of a quantum neural network model is carried out as follows. At $t=0$ the parameters $\pmb{\theta}$ that will be used in the U parameterization are randomly generated from a probability distribution, for example the Gaussian distribution or the uniform distribution. Afterwards, the training data $\mathcal{D}$ is passed to the model. Using the predicted labels $f(\pmb{x}_{i})$ and the target labels $y_{i}$ we calculate the cost function. The objective is for the model to be able to predict the labels $y_{i}$. However, initially the labels predicted by the model will be different from the target labels. In this case, we use the gradient of the cost function to update the $U$ parameters. After updating the parameters, the $\mathcal{D}$ data is passed back to the model and the cost function is calculated again and the $U$ parameters are optimized. This training process is repeated $t=T$ times.

Several models have already been proposed in the literature. Quantum multilayer perceptron \cite{quantum_multilayer_perceptron}, quantum convolutional neural networks \cite{S_J_Wei}, quantum kernel models \cite{kernel_methods} and hybrid quantum classical neural networks \cite{J_Liu} are some of them. However, all these models follow the same operating principle, that is, training data is passed to a quantum circuit that runs on a quantum computer and a classical optimizer is used to update the model parameters in order to minimize the cost function. 


\section{On the effect of the variational quantum circuit parametrization}
\label{sec:thm}

Our main objective here is to answer the question: Does the architecture of the quantum chips significantly influence the efficiency of a quantum machine learning model? As we discussed in the previous section, the parametrization in general can be written in the form given by Eq. \eqref{eq:parametrization_1}, with each $U_{l}$ given by Eq. \eqref{eq:parametrization_2}. Initially, we observe that the parametrization $U_{l}$ given by Eq. \eqref{eq:parametrization_2} does not depend on the architecture of the chip that will be used, since each one of its quantum gates acts on a single qubit. In fact, any gate on a qubit is independent of the chip architecture. However, the parameterization $W_{l}$, obtained when using gates that do not depend on the parameters, as for example the CNOT, SWAP, CZ, and CY gates, will depend on the architecture of the quantum chip.

As we mentioned previously, the chip architecture will influence our parameterization. Depending on the architecture, the depth of the parameterization will be seriously influenced due to the need to apply several SWAP ports. However, although we previously discussed this issue only for CNOT gates, this same analysis applies to any gate that acts on more than one qubit. In fact, the more restricted the chip connectivity is, the deeper the parameterization will be. 
Thus, we are naturally led to ask whether for quantum machine learning models it is indeed necessary to make use of these non-parameterized quantum gates involving qubits that do not interact directly. To carry out this analysis, we begin by defining the cost function as
\begin{equation}
    C = \frac{1}{D}\sum_{i=1}^{D}( f(x_{i},U) - y_{i} )^{2},\label{eq:loss_function}
\end{equation}
where $D$ is the size of our training dataset, $f(x_{i},U)$ is the quantum model output defined in Eq. \eqref{eq:medida} with the parametrization $U$ given by Eq. \eqref{eq:parametrization_1} and $y_{i}$ is the desired output given the input $x_{i}$. In addition, to prove the results in the sequence we use some properties of $t$-designs, that are defined as follows.

\begin{definition}
A unitary $t$-design is defined as a finite set $\{ W_{y}\}_{y \in Y} $ (of size $|Y|$) of unitary matrices $W_{y}$ on a $d$-dimensional
Hilbert space such that for every polynomial $P(t,t)(W)$ of degree at most $t$ in the matrix elements of $W$, or of $W^{\dagger}$ \cite{Dankert_Christoph}, we have
\begin{equation}
    \frac{1}{|Y|}\sum_{y \in Y}P(t,t)(W_{y}) = \int_{U(d)} d\mu(W) P(t,t)(W),
\end{equation}
where in the right-hand side $U(d)$ denotes the unitary group of degree $d$ and $\mu(W)$ is the Haar measure over $U(d)$.
\end{definition}
This last equation implies that the average of $P(t,t)(W)$ is indistinguishable from the integral over $U(d)$ with an uniform-Haar measure.
For more details we suggest Ref. \cite{Symbolic_integration}. With this we can state the following theorem.

\begin{theorem}
\label{tr:1}
Let the cost function be defined in Eq. \eqref{eq:loss_function}, let the training dataset be $\mathcal{D} := \{x_{i},y_{i}\}_{i=1}^{D}$, and let the parametrization be given by Eq. \eqref{eq:parametrization_1}. We have that the average value of the cost function over the parametrizations $U$ will be limited by
\begin{equation}
     \mathbb{E}_{U}[C] \geqslant  \frac{1}{D}\sum_{i=1}^{D} Var_{U}[f(x_{i},U)]. \label{eq:med_1}
\end{equation}
\end{theorem}

\begin{proof}
To prove this theorem, we start by rewriting the cost function in Eq. \eqref{eq:loss_function} as
\begin{equation}
     C = \frac{1}{D}\sum_{i=1}^{D}( f(x_{i},U) - \mu + \mu - y_{i} )^{2},
     \label{eq:loss_function_2}
\end{equation}
where $\mu := \mathbb{E}_{U}[f(x_{i},U)] $. So we have
\begin{equation}
\begin{split}
    \mathbb{E}_{U}[C] =& \frac{1}{D}\sum_{i=1}^{D} \mathbb{E}_{U}[( f(x_{i},U) - \mu + \mu - y_{i} )^{2}] \\
                        =& \frac{1}{D}\sum_{i=1}^{D} \mathbb{E}_{U}[ ( f(x_{i},U) - \mu)^{2} + \\
                        & + 2( f(x_{i},U) - \mu)(\mu - y_{i} ) + (\mu - y_{i}  )^{2}].
                        \label{eq:loss_function_3}
\end{split}
\end{equation}
Since $\mathbb{E}_{U}[2( f(x_{i},U) - \mu)(\mu - y_{i} )] = 0 $, then Eq. \eqref{eq:loss_function_3} reduces to
\begin{align}
    \mathbb{E}_{U}[C] &= \frac{1}{D}\sum_{i=1}^{D} \mathbb{E}_{U}[ ( f(x_{i},U) - \mu)^{2}  + (\mu - y_{i}  )^{2}] \nonumber \\
 & = \frac{1}{D}\sum_{i=1}^{D} \mathbb{E}_{U}[ ( f(x_{i},U) - \mu)^{2}]  + \mathbb{E}_{U}[(\mu - y_{i}  )^{2}] \nonumber \\
    & \geqslant \frac{1}{D}\sum_{i=1}^{D} Var_{U}[f(x_{i},U)].
    \label{eq:loss_med_1}
\end{align}
With this we complete the proof of Theorem \ref{tr:1}.
\end{proof}

From this Theorem, we can prove the following corollary.
\begin{corollary}
\label{cor:1}
Let the cost function be defined by Eq. \eqref{eq:loss_function}, let the training dataset be $\mathcal{D} := \{x_{i},y_{i}\}_{i=1}^{D}$, and let the parametrization to be given by the Eq. \eqref{eq:parametrization_1}. Then the deviation of the cost function is bounded as follows:
\begin{equation}
      |C-\mathbb{E}_{U}[C]| \leqslant \frac{d+1}{d^{2}(d^{2}-1)}\bigg( |Tr[H]^{2}| + d|Tr[H^{2}]| \bigg) + |C|,
\end{equation}
with $d=2^{N}$, where $N$ is the number of qubits used in the VQA.
\end{corollary}

\begin{proof}
To prove this Corollary, we must first calculate the variance of Eq. \eqref{eq:medida}, which is defined as $Var_{U}[f(x_{i},U)] = \mathbb{E}_U\big[(f(x_{i},U))^{2}\big] - \big(\mathbb{E}_U [f(x_{i},U)]\big)^{2} $. 
If $U$ is $1$-design, then we have \cite{BR_cost_Dependent}
\begin{equation}
    \mathbb{E}_U[f(x_{i},U)] = \int d\mu(U)Tr[U\rho_{i} U^{\dagger}H] = \frac{Tr[H]}{d}.
    \label{eq:int_1}
\end{equation}
Above we use $Tr[\rho_{i}] = 1 \ \forall i $. 
If $U$ is a $2$-design, then it follows that \cite{BR_cost_Dependent}
\begin{align}
     \mathbb{E}_U\big[(f(x_{i},U))^{2}\big] & = \int d\mu(U)Tr[U\rho_{i} U^{\dagger}H]Tr[U\rho_{i} U^{\dagger}H] \nonumber \\
    & = \frac{Tr[H]^{2}+Tr[H^{2}]}{d^{2}-1}\bigg( 1 - \frac{1}{d}\bigg),
    \label{eq:int_2}
\end{align}
where we used $Tr[\rho_{i}^{2}]=1$,  since $\rho_i$ are pure states. So, using Eqs. \eqref{eq:int_1} and \eqref{eq:int_2} we get
\begin{equation}
    Var_{U}[f(x_{i},U)] = \frac{1-d}{d^{2}(d^{2}-1)} Tr[H]^{2} + \frac{d-1}{d(d^{2}-1)}Tr[H^{2}]. \label{eq:var}
\end{equation}

From Theorem \ref{tr:1} and from the triangle inequality for complex numbers, we obtain
\begin{align}
\big|C-\mathbb{E}_{U}[C]\big| & \le \Big|C-\frac{1}{D}\sum_{i=1}^{D} Var_{U}[f(x_{i},U)]\Big| \\
& \leqslant |C| + \Big|\frac{1}{D}\sum_{i=1}^{D} Var_{U}[f(x_{i},U)]\Big| \\
& \leqslant |C| + \frac{1}{D}\sum_{i=1}^{D} \big|Var_{U}[f(x_{i},U)]\big|.
\end{align}
Now, using the variance obtained in Eq. \eqref{eq:var}, we get
\begin{equation}
      |C-\mathbb{E}_{U}[C]| \leqslant \frac{d-1}{d^{2}(d^{2}-1)}\bigg( \big|Tr[H]^{2}\big| + d\big|Tr[H^{2}]\big| \bigg) + |C|.\label{eq:corP}
\end{equation}
With this, we complete the proof of the corollary.
\end{proof}

To analyze Eq. \eqref{eq:corP}, let us consider the following scenario. We create a quantum neural network model to deal with a binary classification problem where we previously define the number of qubits, the observable $H$ and the parameterization $U$ that will be used by the model. The left side of the inequality tells us the difference between the value of the cost function using the parameterization $U$ with the average value $C$ obtained using different parameterizations. As the average value is obtained over all possible parameterizations, this value includes the behavior of $C$ for parameterizations that do not use any restrictions as well as those that use restrictions. In turn, the right hand side of the inequality tells us an upper limit for this deviation. As the number of qubits and the observable are previously defined, it follows that the first term on the right side of the inequality is simply a constant, which we will call $\delta$. The second term that appears in this inequality is the cost function itself. For the scenario we are considering when initializing the parameters, in general, we will obtain a value $C \neq 0$, as the labels predicted by the model will be different from the true labels due to the fact that the parameters used in the parameterization were obtained randomly. However, as we train the model we expect $C$ to go to zero or close to zero. Therefore, we see that the upper bound for the difference of $C$ with its average value $\mathbb{E}_{U}[C]$ will decrease as we train the model. For the case where $\delta \rightarrow 0$ when $d \rightarrow \infty $ and $C \rightarrow 0$ when $T \rightarrow \infty$, with $T$ being the number of epochs, we can see that the right-hand side of the inequality will tend to zero. So, the difference between $C$ and $\mathbb{E}_{U}[C]$ will be zero, and therefore they will be equal. This implies that for any $U$ the behavior of $C$ will be equal to the average behavior over all possible parameterizations regardless of the parameterization $U$ used. This implies that we can use parametrizations constrained by the chip architecture, without significantly losing efficiency in theory and avoiding SWAP gates, improving thus the obtained results in practice.

At this point we must make the following observation. The inequality tells us the upper limit of the difference between $C$ and $\mathbb{E}_{U}[C] $. This does not imply that initially, given two parameterizations $U_{A}$ and $U_{B}$, the values of the cost function will be equal. To illustrate, consider that the upper limit is $0.4$ and $\mathbb{E}_{U}[C] = 0.1$, in this case we can obtain $C(U_{A}) = 0.1$ and $C(U_{B }) = 0.2$. However, as the VQAs are trained, the upper bound approaches zero and we thus have $C(U_{A}) \thickapprox C(U_{B}) $.

\clearpage
\section{Simulation method}
\label{sec:sim}

To exemplify the behavior predicted in the previous section, two sets of experiments were carried out. In the first set of experiments, the parametrization $U$, given by Eq. \eqref{eq:parametrization_1}, was generated randomly.  In this case, the parameterization was generated as follows. Firstly, we consider that all layers of the $U$ parameterization are obtained using the same set of gates, that is, $U_{l}W_{l} = U_{l'}W_{l'}$ are equal. To obtain $U_{l}$ in Eq. \eqref{eq:parametrization_1} we use the set of rotation gates obtained with $\sigma = \{ X,Y,Z \}$ where $\sigma$ was obtained randomly. In turn, to obtain $W_{l}$ we use the set of two-qubit gates $\{I, CNOT,CY,CZ \}$ where $I$ indicates that the identity gate will be applied. 
We also consider in this case that two-qubit gates can be applied to any pair of qubits. This means that they do not have to be applied only between neighboring qubits. Furthermore, the choice of the gate that will be applied between the pairs of qubits was also made randomly. For this first set of experiments, $300$ different parametrizations were generated. For each of these parametrizations, $20$ training epochs were used. We chosed this value because our objective is to analyze how the choice of the parameterization $U$ affects the result, and not to actually train the model as best as possible, in which case a higher value would be ideal. As will be clear from the results below, with this value of epochs, we were already able to carry out the desired analysis. The optimizer used was Adam \cite{ADAM_o}, with a learning rate of $0.001$. The number of qubits used was varied from $2$ to $10$. The parametrization depth was also varied, where we used $L=2$, $L=4$, and $L=6$. We utilized two observables $H$. The first being defined as $H = \mathbb{I} \otimes \sigma_{Z}$, where the $\sigma_{Z}$ Pauli matrix was always applied to the last qubit of the quantum circuit. We also used the observable $H =  \mathbb{I} \otimes |0 \rangle \langle 0|$,  measuring $|0 \rangle \langle 0|$ also on the last qubit.



For the second set of experiments, we use a fixed parameterization. That is, we use the parameterization given by Eq. \eqref{eq:parametrization_1} with each layer $U_{l}W_{l}$ given by the parameterization shown in Fig. \ref{fig:param_fixo_2}. 
We performed simulations without and with restrictions on the parametrization $U$. The restriction used was that the gates that act on more than one qubit could only be applied to qubits that were directly connected in the quantum chip. To know which qubits would be side by side, IBM's Guadalupe quantum chip was used as a model; see Fig. \ref{fig:chip_guadalupe}. For these experiments, both the number of qubits and the parametrization depth were varied, where the values used were equal to those used in the first set of experiments. Besides, the same optimizer with the same learning rate as in the previous experiments were used here. 

\begin{figure}[b]
    \centering
    \includegraphics[scale=0.4]{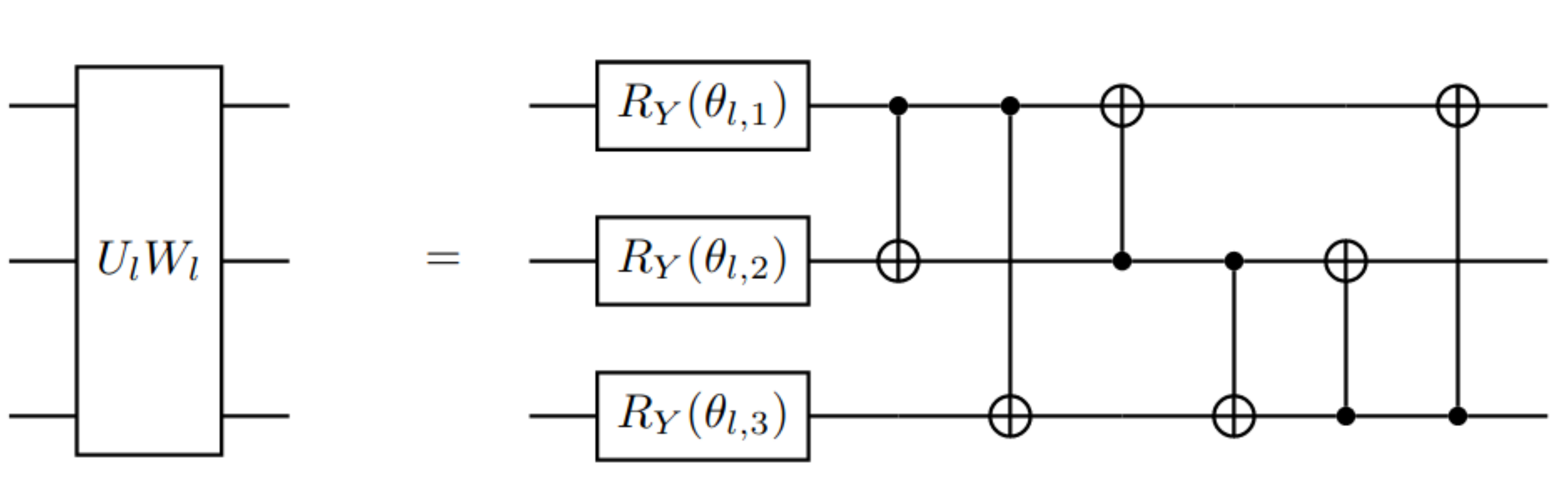}
    \caption{Illustration of the $U_{l}W_{l}$ parametrization for the case where it is not fixed by the chip architecture.}
    \label{fig:param_fixo_2}
\end{figure}

In the first set of experiments, $300$ different parametrizations were generated. We used a high number of parametrizations in order to have greater precision in the results obtained. However, in the second set of experiments, as the parametrization is fixed with the only difference being the initial parameters, we chose to do each experiment 50 times.

\begin{figure}[b]
    \centering
    \includegraphics[scale=0.9]{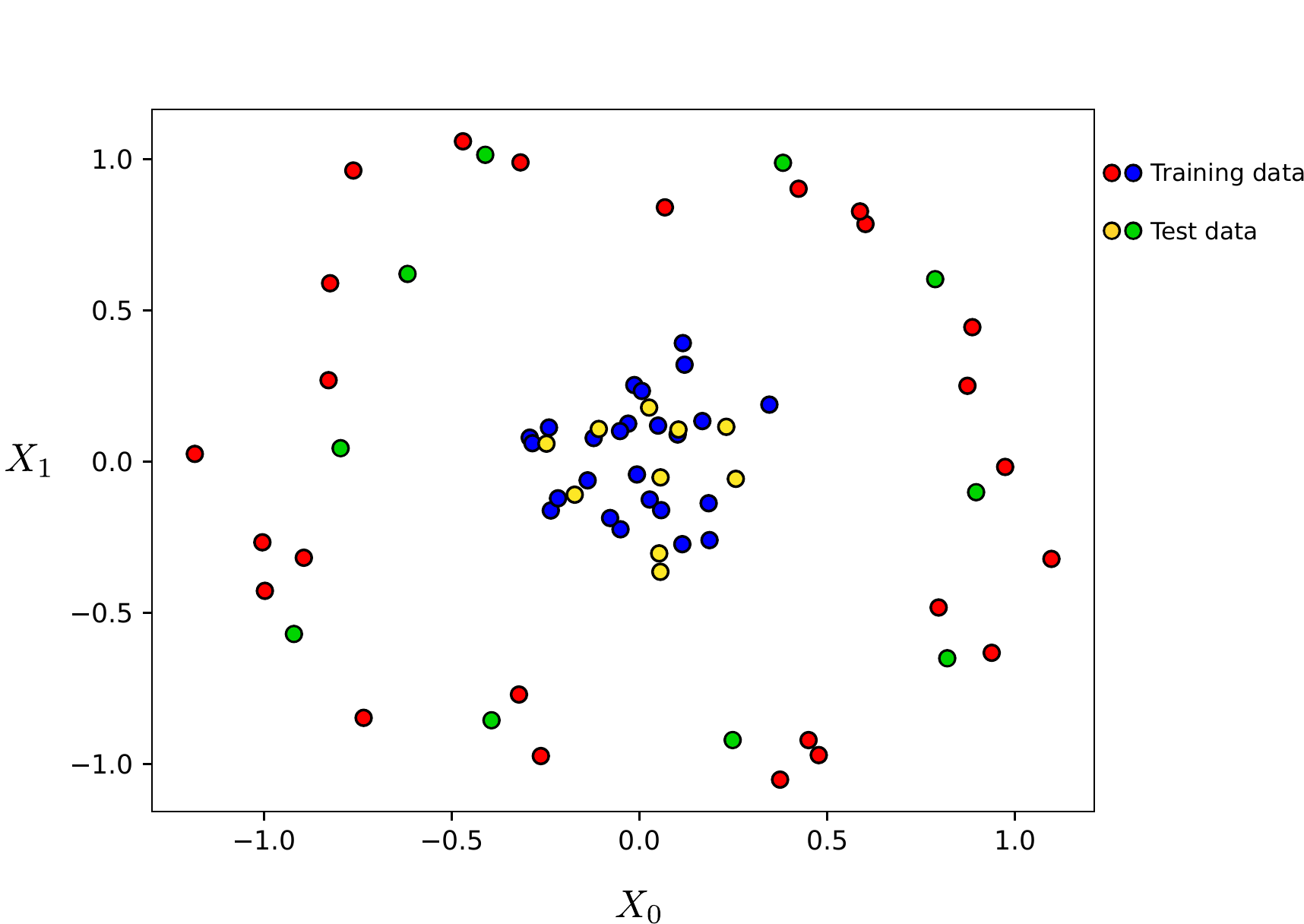}
    \caption{Example dataset used to train and test models of quantum neural networks. The red and blue dots are the training data. The yellow and green dots are test data.}
    \label{fig:data_train}
\end{figure}

Finally, the data used for training the models were obtained using the machine learning library scikit-learn \cite{Scikit_learn_ref}. In Fig. \ref{fig:data_train}, an example of training data used during these experiments is illustrated. The data is divided into training data, points in red and blue, and test data, points in yellow and purple. For this work we only use the training data, as our objective is to analyze the behavior of the cost function. In Fig. \ref{fig:codificacao_data} is illustrated how the encoding of data in a quantum state was done for all experiments. As the input data is composed of vectors of two elements, we choose to encode the two values in all qubits.

\begin{figure}
    \centering
    \includegraphics[scale=0.5]{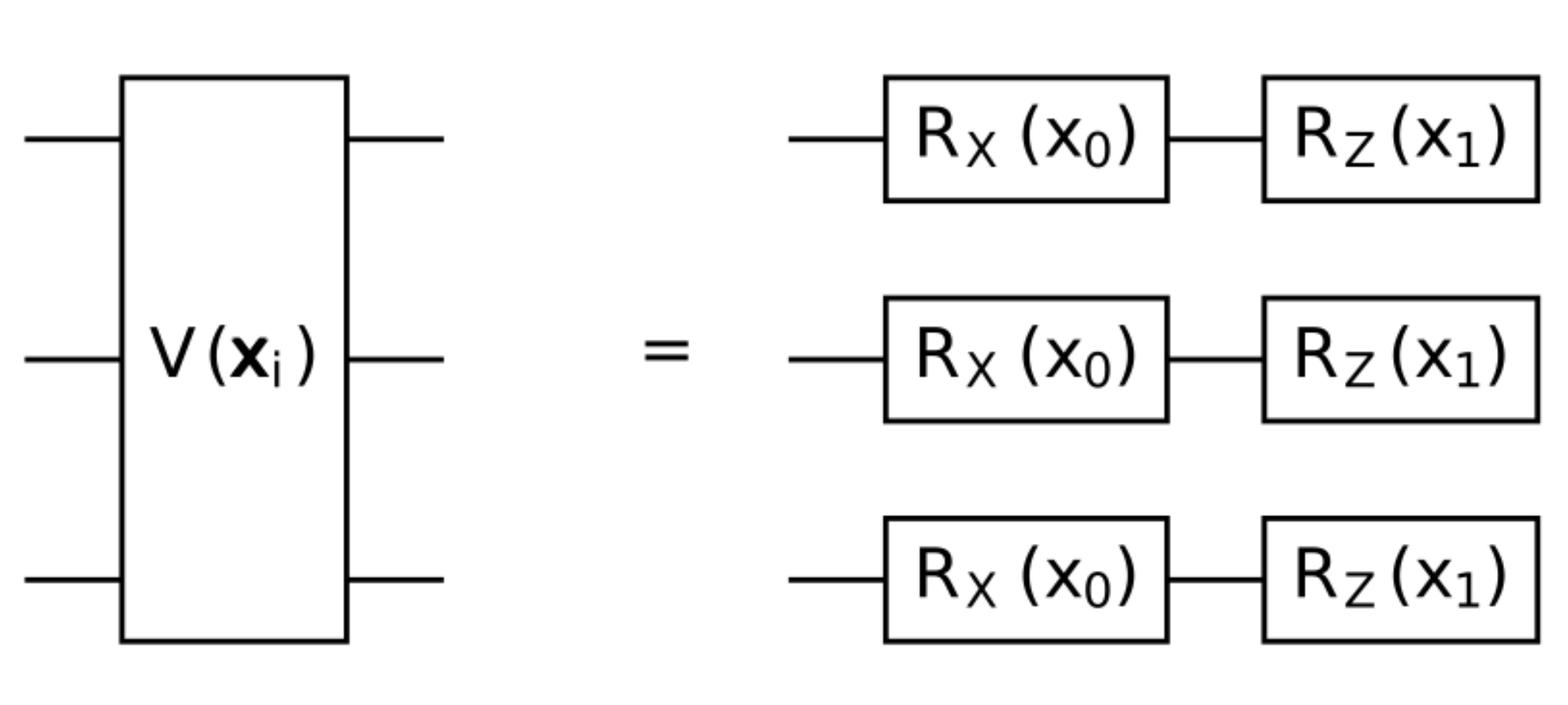}
    \caption{Illustration of the way data was encoded in a quantum state. In this example only three qubits are shown. However, the same encoding was used in all the experiments, independent of the number of qubits involved. In this example we show how input data with two elements $x_{0}$ and $x_{1}$ will be encoded in the quantum circuit.}
    \label{fig:codificacao_data}
\end{figure}



\clearpage
\section{Simulation results}
\label{sec:results}

In this section, we present the results obtained in the numerical simulations. We begin by presenting the results obtained by randomly generating the parameterization $U$, Figs. \ref{fig:Ualeatorio_O_pro} and \ref{fig:Ualeatorio_O_PauliZ}. In Fig. \ref{fig:Ualeatorio_O_pro}, we see the behavior of the cost function for the case where the chosen observable was $H = \mathbb{I} \otimes |0 \rangle \langle 0|$. In these graphs and others, we use loss as a synonym for cost. In this graph, the dark blue lines show the average behavior of the cost function in relation to the 300 parameterizations used, while the light blue lines show the maximum and minimum behavior of the cost function. As we can see, as the size of the system increases, that is, the number of qubits and parameterization depth used increases, the behavior of the cost function approaches an average behavior, especially for the case where the number of qubits $NQ=10$ was used and $L=6$. From Corollary \ref{cor:1}, we already expect this to occur as the number of qubits used increases. However, it never explicitly indicates that this behavior will also depend on the depth of the parameterization. However, we must note that the results obtained in Corollary \ref{cor:1} were only possible because we assumed that the generated parameters lead to a $t-$design. Therefore, we argue that this occurs because, as some results in the literature  indicate \cite{Expressibility_Sim}, in general, the deeper the parameterization the closer it is to a $t-$design. Thus, this behavior is also explained by the Corollary \ref{cor:1}.

In Fig. \ref{fig:Ualeatorio_O_PauliZ}, we see the behavior of the cost function where the observable used was $H = \mathbb{I} \otimes \sigma_{Z}$. Again, the dark blue lines indicate the average behavior of the cost function in relation to the 300 parameterizations, with the light blue shadow indicating the maximum and minimum behavior of the cost function. Again we see that as the system size increases the cost function approaches the average behavior. In these graphs we can also see more clearly the dependence on the cost function. As previously discussed, the concentration of the cost function around an average value will depend on its own value, where the lower its value, the closer it will be to the average value. Thus, as we can see during training, which is done with the aim of reducing the value of the cost function, the concentration of the function around the average value,  in light blue, tends to get closer to the average value.



In the second set of simulations, Figs. \ref{fig:Ufixo_O_PauliZ} and \ref{fig:Ufixo_O_prob}, we use a fixed parameterization where each layer $U_{l}W_{l}$ is obtained from the parameterization illustrated in Fig. \ref{fig:param_fixo_2}. 
Again, we use the observables $H = \mathbb{I} \otimes \sigma_{Z}$,  Fig. \ref{fig:Ufixo_O_PauliZ}, and $H = \mathbb{I} \otimes |0 \rangle \langle 0|$, Fig. \ref{fig:Ufixo_O_prob}. In both graphs we analyze the behavior of the cost function using the parameterization given by Eq. \eqref{eq:parametrization_1} with each layer $U_{l}W_{l}$ obtained from the parameterization illustrated in Fig. \ref{fig:param_fixo_2} in two different scenarios. In the first scenario we analyze the behavior of the cost function when no restrictions are imposed. In the second scenario we analyze the behavior of the cost function when we impose the restriction that CNOT gates can only be applied between pairs of qubits that are directly connected. For this we use as a reference the architecture of the Guadalupe chip shown in Fig. \ref{fig:chip_guadalupe}.

As we can see, the behavior of the cost function using the $U$ parameterization with and without restrictions is similar, especially when the size of the system used is large. Once again we observe a dependence on the depth of the parameterization. However, as we previously mentioned, this is due to the fact that the greater the depth the closer it is to a $t-$design. We also observed that the average value of the cost function differs slightly between the parameterization with and without restriction. This can be explained by the Theorem \ref{tr:1}, where we see that the average value is limited by the variance of the cost function, given by Eq. \eqref{eq:medida}, which depends on the parameterization $U$. Although we have considered that this variance is taken over all possible parameterizations, it can also be interpreted as the variance given a single parameterization $U$ where in this case the variance is computed over the set $\{ U(\pmb{\theta}_{1}),U(\pmb{\theta}_{2}), U(\pmb{\theta}_{3}),... \}$.  Therefore, it is expected that there is a difference between the average value obtained using the parameterization with and without restriction. However, this value should decrease as the system size increases, Eq. \eqref{eq:var}, a fact observed in the numerical simulation results.

\begin{figure*}
    \centering    \includegraphics[scale=0.18]{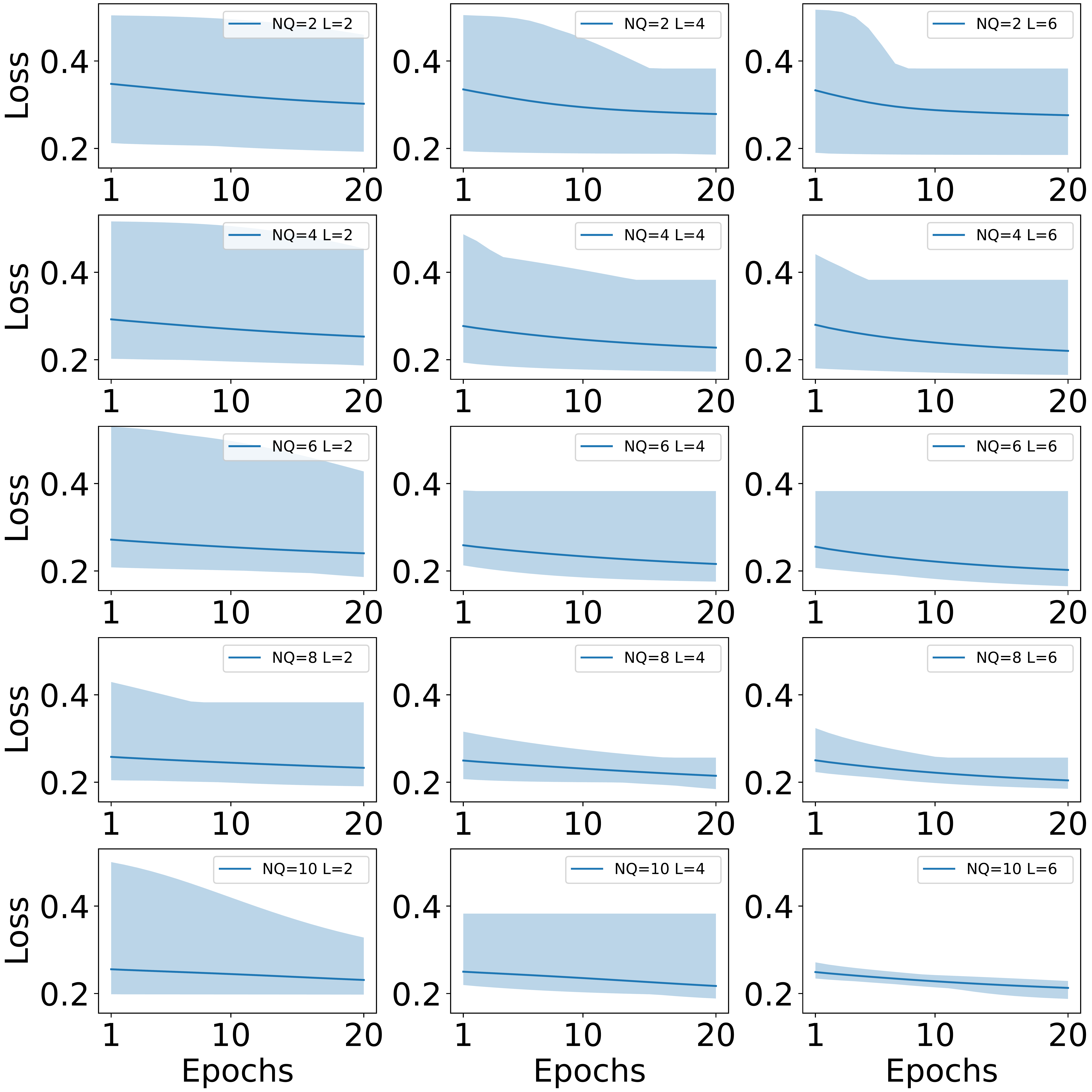}
    \caption{Behavior of cost function during training for a quantum neural network model as the number of qubits increases for different depths of the parametrization. To obtain this graph we randomly generated $300$ parameterizations. In dark blue the average behavior of the cost function during training is shown and in light blue the maximum and minimum behavior is shown. To obtain this graph using 20 epochs. The observable used here was $H = \mathbb{I} \otimes |0 \rangle \langle 0|$ .}
    \label{fig:Ualeatorio_O_pro}
\end{figure*}


\begin{figure*}
    \centering
    \includegraphics[scale=0.18]{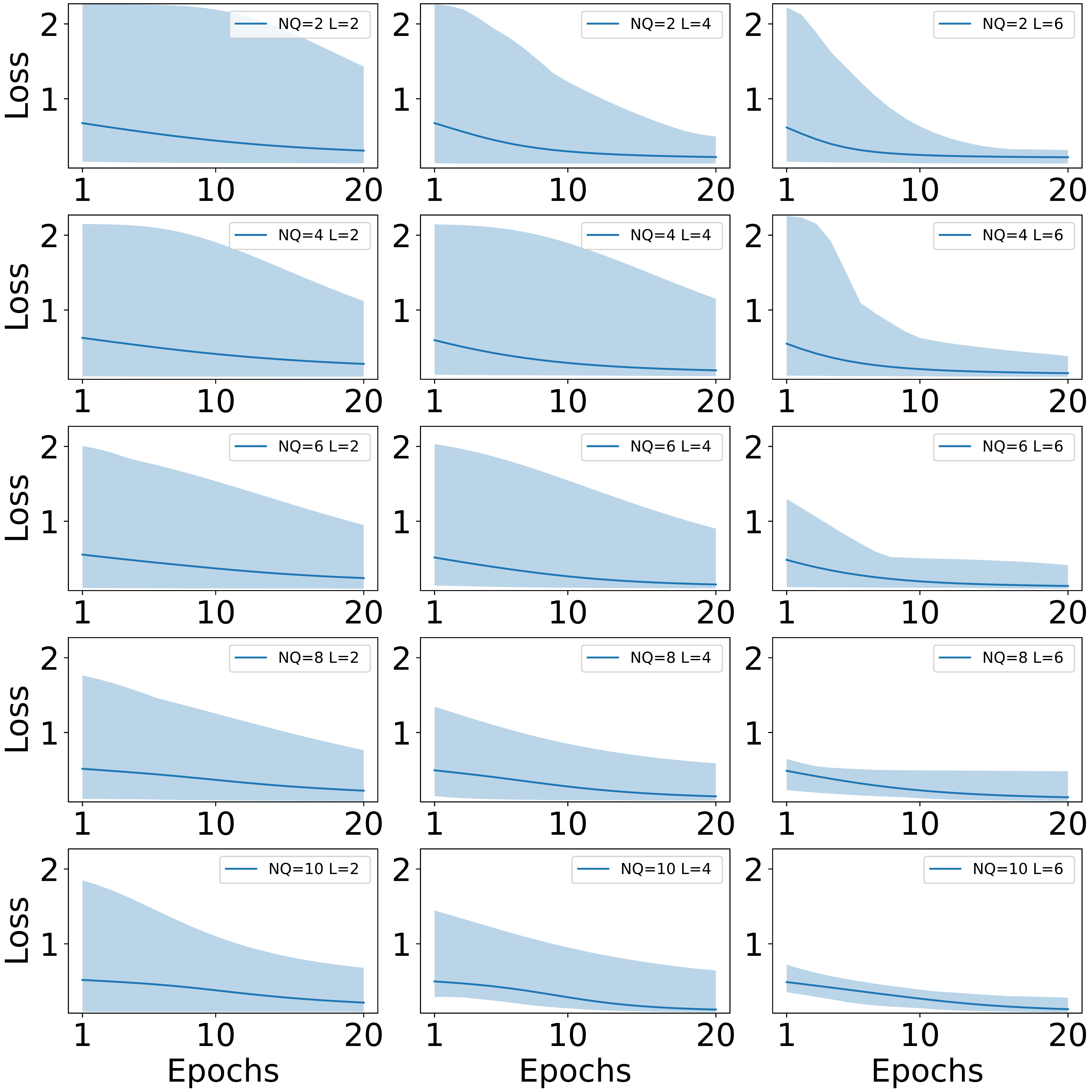}
    \caption{As in the previous graph, we analyze the behavior of the cost function as the number of qubits and depth of the parameterization vary. Again, we generate 300 parameterizations and show the average behavior, dark blue, and the maximum and minimum behavior of the cost function, light blue. The observable used here was $H = \mathbb{I} \otimes \sigma_{Z}$ .}
    \label{fig:Ualeatorio_O_PauliZ}
\end{figure*}


\begin{figure*}
    \centering    \includegraphics[scale=0.18]{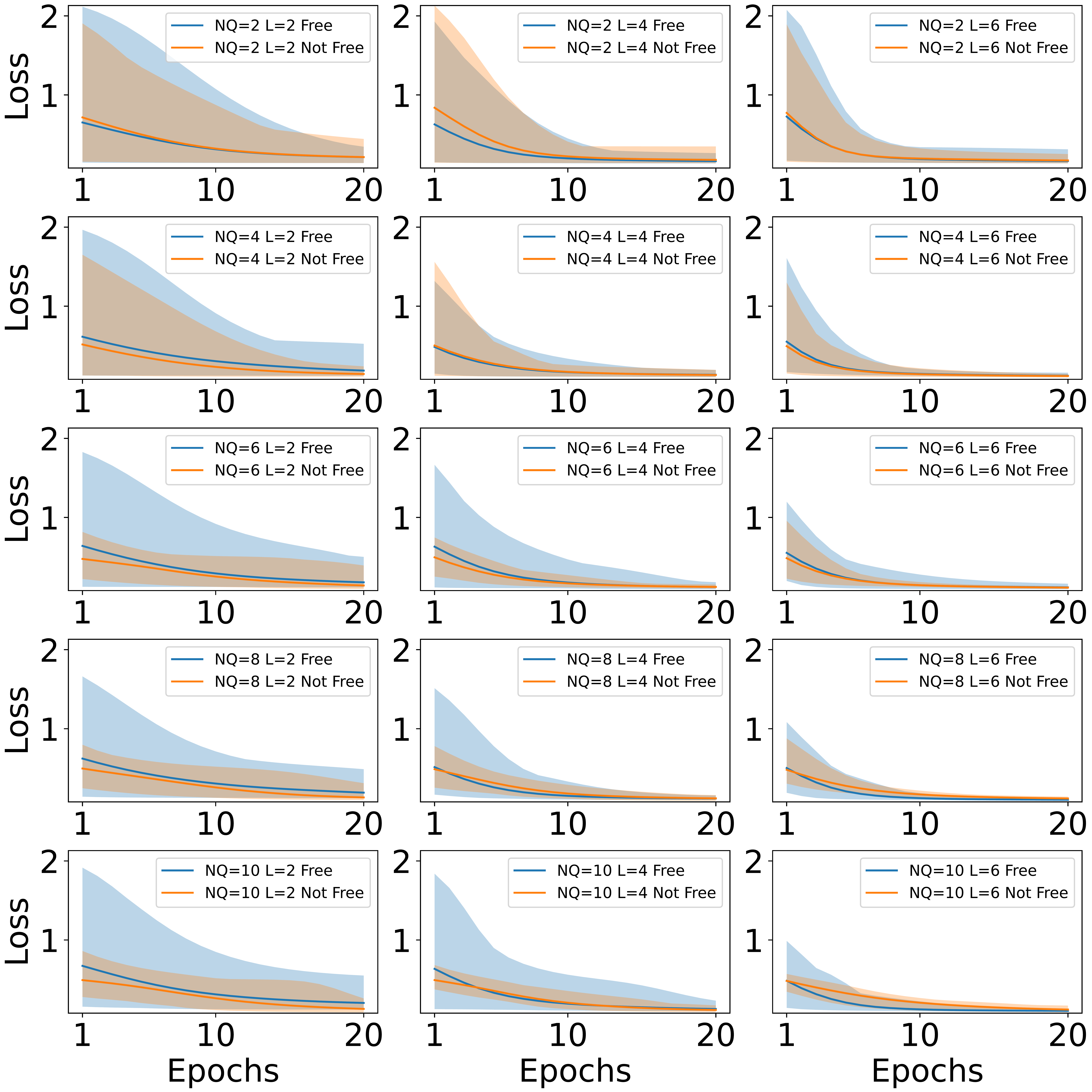}
    \caption{Behavior of a quantum neural network model during training as the number of qubits increases for different parameterization depths. In these numerical simulations, we use the parameterization shown in Fig. \ref{fig:param_fixo_2} with restrictions, which we call Not Free, and without restrictions, which we call Free. The imposed restriction is that CNOT gates can only be applied to pairs of qubits that are directly connected. To do this, we use the Guadalupe chip architecture as a reference. To train this parameterization we used $20$ epochs. For the figures, we use the observable $H = \mathbb{I} \otimes \sigma_{Z}$. For the same parameterization, we performed the simulations 50 times to analyze the parameter initialization dependence.}
    \label{fig:Ufixo_O_PauliZ}
\end{figure*}


\begin{figure*}
    \centering    \includegraphics[scale=0.18]{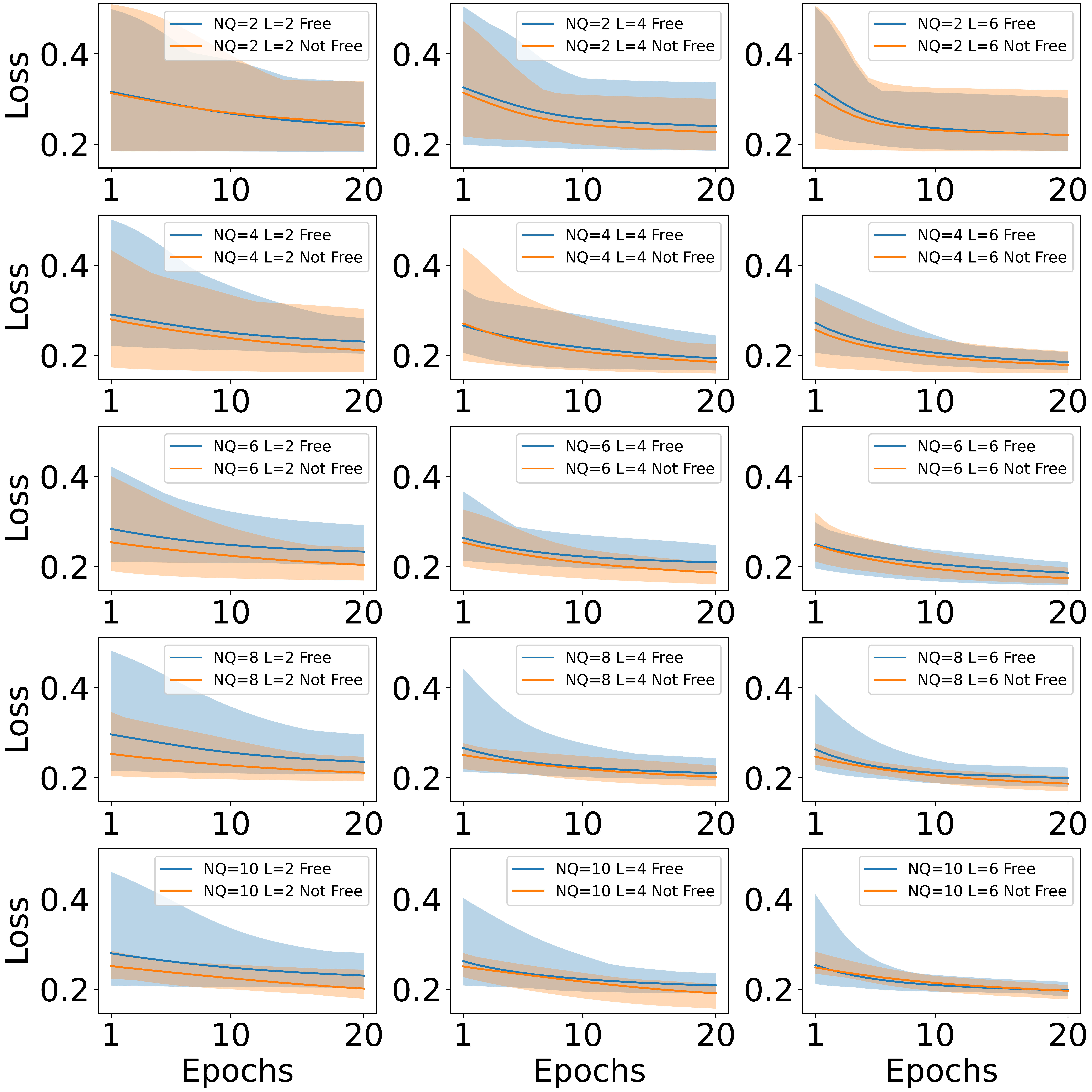}
    \caption{Behavior of a quantum neural network model during training as the number of qubits increases for different parameterization depths. In these numerical simulations, we use the parameterization shown in Fig. \ref{fig:param_fixo_2} with restrictions, which we call Not Free, and without restrictions, which we call Free. The imposed restriction is that CNOT gates can only be applied to pairs of qubits that are directly connected. To do this, we use the Guadalupe chip architecture as a reference. To train this parameterization we used $20$ epochs. For the figures, we use the observable $H = \mathbb{I} \otimes |0 \rangle \langle 0|$. For the same parameterization, we performed the simulations $50$ times to analyze the parameter initialization dependence.}    \label{fig:Ufixo_O_prob}
\end{figure*}

\clearpage

\section{Conclusions}
\label{sec:conc}

In this article, we address a fundamental query: Does the architecture of quantum chips impact the outcomes of a quantum machine learning model? To shed light on this, we derive significant theoretical insights, encapsulated in Corollary \ref{cor:1}, which suggest that the cost function $C$ tends to converge toward an average value, denoted as $E_U[C]$. This tendency is contingent upon several factors, notably the number of qubits employed in the circuit, the observable being measured, and the nature of the cost function itself. Our numerical simulations validate the predictions presented in Corollary \ref{cor:1}. Notably, we observe a discernible relationship between the depth of parameterization and the results from our numerical experiments. However, as discussed earlier, this connection arises due to the assumption made in Corollary \ref{cor:1}, considering the parameterization $U$ as a $t$-design. Prior research, Ref. \cite{Expressibility_Sim}, establishes that as the depth increases, the parameterization progressively aligns with a $t$-design. Hence, the observed behavior, though not explicitly stated in Corollary \ref{cor:1}, finds its explanation in this contextual understanding.

In Figs. \ref{fig:Ufixo_O_PauliZ} and \ref{fig:Ufixo_O_prob}, under a fixed parameterization $U$, we conducted a comparative analysis of the cost function's behavior when imposing specific restrictions on the parameterization against a scenario with no imposed restrictions. The restriction in question involved permitting $CNOT$ gates to exclusively act on physically connected qubit pairs, using IBM's Guadalupe chip architecture as a benchmark. Remarkably, our numerical findings indicated a general similarity in the behavior of the cost function, with and without imposed restrictions, thus affirming the theoretical underpinnings.

Our focus primarily rested on utilizing the Guadalupe chip architecture as a reference due to the observed behavior illustrated in Figs. \ref{fig:Ualeatorio_O_pro} and \ref{fig:Ualeatorio_O_PauliZ} with randomly generated parameterizations. Consequently, we extrapolate our conclusions to encompass various chip architectures. This includes other chip types like trapped ion chips, photonic chips, and more, acknowledging that each may introduce specific noise impacting the model. Notably, our work suggests the potential removal of gates acting on multiple qubits, particularly the controlled gates $\{CNOT,\ CY,\ CZ,\ \cdots\}$, during quantum machine learning model creation for enhanced efficiency.

It is vital to acknowledge that every chip type carries specific noise characteristics influencing model outcomes. As such, deviations from the theoretical results are expected for each chip type. However, with effective noise mitigation, we anticipate the results aligning closely with theoretical expectations. Consequently, our findings suggest that with a sufficiently expansive parameterization, adhering to chip constraints can effectively guide the creation of a quantum machine learning model.


\begin{acknowledgments}
This work was supported by the National Council for Scientific and Technological Development (CNPq), Grants No. 309862/2021-3, No. 409673/2022-6, and No. 421792/2022-1, and by the National Institute for the Science and Technology of Quantum Information (INCT-IQ), Grant No. 465469/2014-0.
\end{acknowledgments}

\vspace{0.3cm}
\textbf{Data availability.}
The numerical data generated in this work is available at \url{https://github.com/lucasfriedrich97/Restricting_to_the_chip_architecture_maintains_the_quantum_neural_network_accuracy} 

\vspace{0.3cm}

\textbf{Conflict of interest}  The authors declare that they have no known competing financial interests or personal relationships that could have appeared to influence the work reported in this paper.

\end{document}